\documentclass[journal]{IEEEtran}

\ifCLASSINFOpdf
\else
   \usepackage{graphicx}
\fi
\usepackage{url}

\usepackage{graphicx}
\usepackage{amsmath,amsthm,amssymb}
\usepackage{booktabs}
\usepackage{xcolor}
\usepackage{nicefrac}
\usepackage{cancel}
\usepackage{balance}
\usepackage{braket}
\usepackage{hyperref}
\usepackage{microtype}
\usepackage[full]{textcomp}
\newcommand{\sat}[1]{{\color{black}#1}}

\makeatletter
\theoremstyle{plain}
\@ifundefined{theorem}{\newtheorem{theorem}{Theorem}}{}
\@ifundefined{lemma}{\newtheorem{lemma}{Lemma}}{}
\@ifundefined{corollary}{\newtheorem{corollary}{Corollary}}{}
\@ifundefined{proposition}{\newtheorem{proposition}{Proposition}}{}
\theoremstyle{definition}
\@ifundefined{definition}{}{}
\theoremstyle{remark}
\@ifundefined{remark}{}{}
\makeatother

\title{Proving the Limits of Quantum Power Flow}
\author{
Cameron Khanpour$^\star$ and Samuel Talkington$^\sharp$
\thanks{\color{black}
$\star$: School of Electrical and Computer Engineering, Georgia Institute of Technology, Atlanta, GA, USA.  Email: \href{mailto:khanpour@gatech.edu}{khanpour@gatech.edu}.
}
\thanks{\color{black}
$\sharp$: Department of Electrical Engineering and Computer Science, University of Michigan, Ann Arbor, MI, USA. Email: \href{mailto:talks@umich.edu}{talks@umich.edu}.
}
}%
\IEEEaftertitletext{\vspace{-2\baselineskip}}
\begin{document}

\maketitle

\begin{abstract}
This letter proves realistic grid properties limit the applicability of quantum computers for power flow. Grids that split into two large regions meeting at only a few buses, common in transmission networks, force the pseudo condition number of the DC susceptance matrix to grow polynomially in the network size, and long chains of lines bridging such regions force quadratic growth. This rigorously verifies the empirical  results of recent work. We also show that the theory holds without model information with high probability for \sat{independent bounded random line susceptances}. Combined with query and tomography lower bounds, this precludes end-to-end quantum advantage for DC power flow at every readout level, and these obstructions persist through AC power flow, optimal power flow, and unit commitment. All proofs are formally verified in Lean 4.
\end{abstract}
\vspace{-1em}
\section{Introduction}
\label{sec:intro}

Quantum linear system (QLS) solvers~\cite{hhl2009,costa2022} prepare a state proportional to $A^{-1}\ket{b}$ in time polylogarithmic in the dimension of $A$, inspiring a sizable quantum power flow literature advertising exponential speedups, e.g.,~\cite{feng2021}. The generic caveats of QLS subroutines of input loading, output readout, and condition number dependence~\cite{dalzell2023} were recently examined for DC power flow (DCPF) in~\cite{pareek2024}. They found end-to-end costs favor classical iterative methods, and practical \sat{HHL} advantage requires readout level $D$ and condition number $\kappa$ with $D\kappa^{3/2}\ll n/d$ on an $n$-bus system of row sparsity $d$. However, this relies on an \emph{empirical} fit $\kappa \sim n^2$ over PGLib test systems~\cite{pglib2019}. This insightful result leaves open the question as to whether it applies to general classes of grid structures. 

To this end, this letter asks the following question:
\begin{quote}\centering
\emph{Is the ill-conditioning of power flow systems an inherent consequence of transmission network topology that quantum algorithms cannot avoid?}
\end{quote}
We prove this in the affirmative, and propagate the consequences through adjacent problems in power engineering. Furthermore, every theoretical result is machine-checked in Lean~4\footnote{\url{https://github.com/eigenergy/quantum-power-flow-limits}}; the first power systems contribution to do so. Section~\ref{sec:conditioning} shows that \emph{small balanced separators}\textemdash which are implied by $\tau$-bounded treewidth or near planarity, both of which are common characteristics of realistic transmission grids~\cite{madani2015,birchfield2017}\textemdash force the pseudo condition number of a grid with maximum degree~$\Delta$ to $\Omega(n/\Delta\tau)$, with transfer corridors forcing $\Omega(n^2)$, explaining the scaling in~\cite{pareek2024}. We then show that both bounds hold with high probability for any \sat{independent bounded random susceptances} given a fixed topology. Section~\ref{sec:e2e} adds query and tomography lower bounds to rule out end-to-end advantage at \textit{every} readout level \sat{$D\geq1$ in the sense of~\cite{pareek2024}}. Section~\ref{sec:beyond} discusses how the obstructions may impact AC power flow (ACPF), DC optimal power flow (OPF), AC-OPF, security constrained OPF, unit commitment (UC), and hybrid schemes.
Section~\ref{sec:outlook} concludes that nearly-linear time Laplacian solvers, randomized numerical linear algebra, and \textit{quantum-inspired} algorithms may help deliver the speedups sought from quantum hardware.

\section{Ill-Conditioning is Structural}
\label{sec:conditioning}

Let $G = (V, E)$ be a connected graph with $n=|V|$ buses, $m=|E|$ branches with series susceptances $b_e > 0$, and maximum degree~$\Delta$. The DCPF problem is $B\theta = p$, $p \perp \mathbf{1}$, with susceptance Laplacian $B = A^\top\!\operatorname{diag}(b_e)A \succeq 0$, $\ker B = \operatorname{span}\{\mathbf{1}\}$ for incidence matrix $A$. For $F \subseteq E$ write $b(F) \triangleq \sum_{e\in F} b_e$, let $\bar{b} \triangleq b(E)/m$, $b_{\max} \triangleq \max_e b_e$, and let $\partial S$ be the edges leaving $S \subset V$. The \emph{pseudo condition number} is $\kappa_+(B) \triangleq \lambda_{\max}(B)/\lambda_2(B)$, with $\lambda_2$ the smallest nonzero eigenvalue. All bounds transfer to the grounded matrix $B_r$ (slack bus deleted), which solvers invert. 
Poincar\'e separation gives $\lambda_{\min}(B_r) \leq \lambda_2(B)$, and \sat{$\operatorname{tr}(B_r)=2b(E)-b(\partial\{r\})$, where $b(\partial\{r\}) \leq \Delta \cdot b_{\max}$ and $b(\partial\{r\})=o(b(E))$ whenever $\Delta \cdot b_{\max}=o(b(E))$}.

Three computational facts anchor the analysis. \emph{(F1)}~\sat{Every rigorous fault tolerant QLS solver in the fixed matrix oracle model, with controlled input preparation, uses $\Omega(\kappa)$ queries on some input, matched by~\cite{costa2022}.} \emph{(F2)}~Recovering a classical $\varepsilon$-approximation of an $n$-dimensional pure state requires $\widetilde{\Theta}(n/\varepsilon)$ applications of its preparation unitary \sat{and inverse}~\cite{apeldoorn2023}, each a fresh QLS solve, and loading a dense classical vector costs $\Omega(n)$ gates per repetition absent QRAM~\cite{pareek2024}. \emph{(F3)}~Classically, $B$ is symmetric diagonally dominant, solvable in $\widetilde{O}(m\log(1/\varepsilon))$ time~\cite{kyng2016, kyng2026}, where $\widetilde{O}$ hides $\operatorname{polylog}(n)$ factors. Combinatorial preconditioning makes the $\kappa$-dependence logarithmic.

\begin{lemma}[Weighted cuts control conditioning]
\label{lem:cut}
Let $G$ be connected. Then (i) $\lambda_{\max}(B) \geq \max\{2b(E)/(n-1),\, 2\max_e b_e\}$, and (ii) for every nonempty proper $S \subset V$ with $\gamma \triangleq |S|/n$, $\lambda_2(B) \leq b(\partial S)/(\gamma(1-\gamma)n)$. Consequently,
\begin{equation}
\label{eq:cutbound}
    \kappa_+(B)
    \geq
    \frac{2\gamma(1-\gamma)}{b(\partial S)}\,
    \max\Bigl\{\, b(E),\; n \max_{e\in E} b_e \Bigr\}.
\end{equation}
\end{lemma}

\begin{proof}
(i) $\operatorname{tr}(B) = 2b(E)$ spreads over at most $n-1$ nonzero eigenvalues, and for $e = \{i,j\}$ the unit vector $y = (\mathbf{e}_i - \mathbf{e}_j)/\sqrt{2}$ has $y^\top B y \geq 2b_e$. (ii) $x = \mathbf{1}_S - \gamma\mathbf{1} \perp \mathbf{1}$ satisfies $x^\top B x = b(\partial S)$ and $x^\top x = \gamma(1-\gamma)n$, then apply Rayleigh--Ritz on $\mathbf{1}^\perp$. Combine with $n/(n-1)\geq 1$ for the result.
\end{proof}

The weighted form~\eqref{eq:cutbound} is strictly stronger than a homogenized bound from~$b_{\min}A^\top A \preceq B \preceq b_{\max}A^\top A$, as $b(E)/b(\partial S)$ compares total network stiffness to \emph{interface} stiffness. A single electrically weak tie line forces ill-conditioning even when the spread $b_{\max}/b_{\min}$ renders the homogenized bound vacuous.

Transmission networks supply such cuts on two grounds. Combinatorially, tree decompositions of width far below system size are routinely computed for large-scale grids, the structure behind chordal relaxations~\cite{madani2015, jabr2012}. Electrically, Lemma~\ref{lem:cut} needs only a balanced cut with $b(\partial S) \ll b(E)$, the ``weak connection'' premise underlying slow coherency. This is consistent with the inter-area oscillations observed in large interconnections~\cite{chow1982}, as $B$ is also the synchronizing torque Laplacian of the swing dynamics. Separations certify such cuts combinatorially. 
We say $G$ admits an \emph{$(s,\beta)$-separation}, $\beta \in (0,\tfrac12]$, if $V = V_A \sqcup X \sqcup V_B$ where $X$ is  the set of buses whose removal would disconnect the graph into two parts with no $V_A$--$V_B$ edges, $|X| \leq s$, and $\min\{|V_A \cup X|, |V_B|\} \geq \beta n$.

\begin{theorem}[Separators force polynomial conditioning]
\label{thm:separator}
If $G$ has maximum degree $\Delta$ and an $(s,\beta)$-separation, then
$$\kappa_+(B) \geq \sat{\frac{2\beta(1-\beta)}{s\Delta}\max\!\left\{\frac{b(E)}{b_{\max}},n\right\}}.$$
\end{theorem}

\begin{proof}
Take $S = V_A \cup X$. Every edge of $\partial S$ has an endpoint in $X$, so $b(\partial S) \leq s\Delta b_{\max}$, and $\gamma(1-\gamma) \geq \beta(1-\beta)$ since $\gamma \in [\beta, 1-\beta]$. Apply Lemma~\ref{lem:cut} for the desired result.
\end{proof}

\begin{corollary}
\label{cor:tw}
(i) If $G$ has treewidth $\tau \leq n/4 - 1$, then $G$ admits a $(\tau+1,\tfrac14)$-separation and $$\kappa_+(B) \geq \sat{\frac{3\max\{b(E)/b_{\max},n\}}{8(\tau+1)\Delta}} = \Omega(n/(\Delta\tau)).$$ (ii) If $G$ is planar and $n \geq 288$, then $$\kappa_+(B) \geq \sat{\frac{5\max\{b(E)/b_{\max},n\}}{18\sqrt{8n}\,\Delta}} = \Omega(\sqrt{n}/\Delta).$$
\end{corollary}

\begin{proof}
(i) A graph of treewidth $\tau$ has a bag $X$, $|X| \leq \tau+1$, whose removal leaves components of size $\leq n/2$ \cite[Ch.~7]{cygan2015}. Greedily group components into $V_B$ until $|V_B| \geq n/4$; then $|V_B| \leq n/2$, as either one component qualifies alone or the last addition was $< n/4$, so $|V_A \cup X| \geq n/2$, and $\sum_i |C_i| \geq \tfrac34 n$ guarantees termination. Apply Theorem~\ref{thm:separator} with $\beta = \tfrac14$, $m \geq n-1$. (ii) Lipton--Tarjan~\cite{liptontarjan1979} gives such a partition with $|V_A|,|V_B| \leq \tfrac{2n}{3}$ and $|X| \leq \sqrt{8n}$. Since $\sqrt{8n} \leq n/6$ for $n \geq 288$, both sides \sat{are at least} $n/6$, a $(\sqrt{8n},\tfrac16)$-separation. Apply Theorem~\ref{thm:separator} with $\beta = \tfrac16$ for the desired result.
\end{proof}

Real grids are often \emph{near} planar rather than planar, but part (ii) persists up to constants \sat{when $c=O(n)$ branch crossings are planarized. Reinvoking~\cite[Thm.~4]{liptontarjan1979}, using nodal costs $1/n,0$ on buses/crossings, projects size $2\sqrt{8(n+c)}=O(\sqrt n)$}.
Either way $\kappa_+$ grows polynomially, defeating polylog$(\kappa)$ hopes. We now show that \sat{quadratic growth can arise from} \emph{corridors}, long interties and transfer interfaces that join regional subnetworks, formalized below.

\begin{proposition}[Corridors imply quadratic conditioning]
\label{prop:corridor}
Suppose $V = V_S \sqcup P \sqcup V_T$ with no $V_S$--$V_T$ edges, where $P = \{p_1,\dots,p_\ell\}$, $\ell \geq 2$, induces a path whose interior vertices have degree two in $G$ and whose endpoints $p_1, p_\ell$ have neighbors only in $V_S \cup \{p_2\}$ and $V_T \cup \{p_{\ell-1}\}$, respectively, and let $E_P$ be the $\ell-1$ path edges. If \sat{$\beta>0$ and} $\min\{|V_S|,|V_T|\} \geq \beta n$, $$\kappa_+(B) \geq \sat{2\beta^2(\ell-1)^2\frac{\max\{b(E),n b_{\max}\}}{b(E_P)} \geq 2\beta^2n(\ell-1)}.$$ Since $b(E_P) \leq (\ell-1)b_{\max}$, \sat{$\ell=\Theta(n)$ and fixed $\beta>0$ give $\kappa_+(B)=\Omega(n^2)$ for all positive weights}. Letting $\ell = 2$ \sat{recovers the tie line scaling} of Theorem~\ref{thm:separator}.
\end{proposition}

\begin{proof}
Let $x_v = -1$ on $V_S$, $+1$ on $V_T$, and $x_{p_i} = -1 + \tfrac{2(i-1)}{\ell-1}$. Only path edges see a nonzero difference, each $2/(\ell-1)$, so $x^\top B x = 4b(E_P)/(\ell-1)^2$. With $\tilde x = x - (\mathbf{1}^\top x/n)\mathbf{1}$, $\tilde x^\top B \tilde x = x^\top B x$, while $\tilde{x}^\top \tilde{x} = \min_{c} \|x - c\mathbf{1}\|^2 \geq \min_{c} [\, |V_S|(1+c)^2 + |V_T|(1-c)^2 ] = 4|V_S||V_T|/(|V_S|+|V_T|) \geq 4\beta^2 n$, using $|V_S||V_T| \geq \beta^2 n^2$ and $|V_S|+|V_T| \leq n$. Hence $\lambda_2(B) \leq b(E_P)/(\beta^2 n(\ell-1)^2)$, then \sat{apply both bounds in Lemma~\ref{lem:cut}(i)}.
\end{proof}

A corridor \sat{gives $\lambda_2(B)=O(b(E_P)/(n\ell^2))$} while the rest of the network ensures $\lambda_{\max} = \Omega(\bar b)$, \sat{providing an explanation for the near quadratic fit $\kappa \sim n^2$ on the PGLib systems~\cite{pareek2024}}. Finally, the bounds need no model information: randomize the branch data to \textit{any} bounded distribution~\cite{talkington2025} to yield the~following.

\begin{proposition}[Samplewise ill-conditioning, random susceptances]
\label{prop:random}
Fix a topology with maximum degree $\Delta$ and an $(s,\beta)$-separation, and let $\{b_e\}$ be independent with $b_e \in (0,\sat{K}]$ a.s.\ \sat{for deterministic support bound $K$} and mean $\bar\mu \triangleq m^{-1}\sum_e \mathbb{E}[b_e]$. Then for every $\varepsilon \in (0,1)$, with probability at least $1 - \exp(-2\varepsilon^2 m\bar\mu^2/\sat{K^2})$, 
$$
\kappa_+(B) \geq 2\beta(1-\beta)(1-\varepsilon)\,\frac{\bar{\mu} m}{Ks\Delta},
$$ and the substitution $b(E) \mapsto (1-\varepsilon)m\bar\mu$ applies, with the same probability, to \sat{the weighted terms in} Corollary~\ref{cor:tw} and Proposition~\ref{prop:corridor}. \sat{The topology only bound holds for every realization under arbitrary joint laws on positive $b_e$.}
\end{proposition}

\begin{proof}
The cut estimates hold samplewise. Only $b(E)$, a sum of $m$ independent variables in intervals of length $\leq \sat{K}$, is random, and Hoeffding gives $\Pr[b(E) \leq (1-\varepsilon)m\bar\mu] \leq \exp(-2\varepsilon^2 m\bar\mu^2/\sat{K^2})$. Apply Theorem~\ref{thm:separator} for the result.
\end{proof}

Boundedness makes each $b_e$ \textit{sub-Gaussian}, so Proposition~\ref{prop:random} covers genericity, uncertainty, or measurement error of $b_e$.

\section{No End-to-End Quantum Advantage for DCPF}
\label{sec:e2e}

Call a sequence of DCPF instances a \emph{grid family} if $\Delta, \tau, s = O(1)$ (hence $m = \Theta(n)$). By Corollary~\ref{cor:tw}, $\kappa = \Omega(n)$, and $\Omega(n^2)$ with macroscopic corridors.

\begin{proposition}[End-to-end lower bounds]
\label{prop:e2e}
On a grid family, for \emph{any} quantum DCPF algorithm built from a QLS solver \sat{with a $p$-independent number of input preparation unitary and inverse calls per solve for fixed $B$}, even granting free (QRAM) state preparation\sat{, with constant success probability}: (i) \sat{for some $p_*\perp\mathbf1$ and $0<\varepsilon<1$,} outputting a classical $\hat\theta$ with $\|\hat\theta - \theta\|_2 \leq \varepsilon\|\theta\|_2$ costs $\widetilde\Omega(n\kappa/\varepsilon) = \widetilde\Omega(n^2/\varepsilon)$ queries, versus $\widetilde O(n\log(1/\varepsilon))$ classically; (ii) estimating \sat{$q_i=(\theta_i/\|\theta\|_2)^2$ or $q_{ij}=((\theta_i-\theta_j)/\|\theta\|_2)^2$ across gap $\gamma>0$ to error $\gamma/4$, using at least one QLS solve,} costs $\Omega(\kappa) = \Omega(n)$ \sat{on one input}, already the price of the classical solve. Without QRAM, (F2) adds $\Omega(n)$ loading gates per repetition \sat{when $p_*$ is dense}. On corridor families the gap in (i) widens to $\widetilde\Omega(n^3/\varepsilon)$ vs. $\widetilde O(n\log(1/\varepsilon))$.
\end{proposition}

\begin{proof}
By (F2) \sat{on the $(n-1)$-dimensional balanced subspace}, a classical output needs $\widetilde\Theta(n/\varepsilon)$ \sat{unitary and inverse calls. For unit modes $u,v$, put $(p^{(0)},p^{(1)})=(\kappa u+v,\kappa u-v)/\sqrt{\kappa^2+1}$. Their input oracles and normalized solutions are $2/\sqrt{\kappa^2+1}$ and $\sqrt2$ apart, so the hybrid argument forces $\Omega(\kappa)$ oracle queries. Reachability via $p_*=B\theta_*/\|B\theta_*\|$ transfers this cost to tomography's hard input; $\gamma/4$ error on a $\gamma$-gap local observable proves (ii).}
\end{proof}

With $\kappa = \Theta(n)$, conjugate gradient costs $\widetilde O(n^{3/2}\log(1/\varepsilon))$ while HHL pays $\widetilde\Omega(n\kappa^2/\varepsilon^2) = \widetilde\Omega(n^{3}/\varepsilon^2)$. Substituting the bound $\kappa = \Omega(n)$ into the advantage criterion~$D\kappa^{3/2} \ll n/d$ of~\cite{pareek2024} forces~$D \ll n^{-1/2} < 1$. No readout level, not even a single scalar, passes on separator structured families. Outputs with provable quantum advantage, such as Forrelation-type observables~\cite{dalzell2023}, remain unlike any power system problem to the best of our knowledge.

\section{Beyond DC Power Flow}
\label{sec:beyond}

\emph{AC power flow.}
\sat{At the lossless unit voltage flat start, the active angle block equals the susceptance Laplacian}~\cite{talkington2025}, and conditioning often degrades in regions away from flat start. Newton outer iterations are also classically adaptive, forming each mismatch from a classical iterate, so a quantum inner solver pays the $\widetilde\Theta(n/\varepsilon)$ tomography cost of (F2) \emph{per iteration}. Nonetheless, very recent results indicate optimism for ACPF, obtaining an~$\Omega(n\kappa/\varepsilon)$ end-to-end baseline~\cite{pareek2026conditionsquantumadvantageac}.

\emph{DC optimal power flow.}
Interior point methods (IPMs) for DC-OPF reduce each iteration to a Newton system whose \sat{positive barrier angle Hessian} is a \emph{weighted} susceptance Laplacian $A^\top\operatorname{diag}(w_{k,e}b_e)A$, $w_{k,e} > 0$. Thus Theorem~\ref{thm:separator} \sat{applies at each strictly feasible iterate}. Quantum IPMs~\cite{augustino2023} must also extract the Newton direction classically at each of $\widetilde O(\sqrt{n})$ iterations, compounding the tomography cost. Classically, the Laplacian structure is an asset; IPMs with fast Laplacian solvers underlie nearly-linear time network flow algorithms~\cite{chen2022}, which may extend to DC-OPF itself.

\emph{AC-OPF, security constraints, and unit commitment.}
AC-OPF feasibility is strongly NP-hard~\cite{bienstock2019}, as is UC~\cite{bendotti2019}. 
A polynomial time quantum algorithm for either \sat{general problem} would imply $\mathrm{NP} \subseteq \mathrm{BQP}$, which is widely believed to not hold. A \emph{particular} quantum algorithm may still solve a \emph{particular} problem instance fast\sat{; however,} it is an open question whether such instances are useful for the power engineering community. Additionally, relative to an \sat{enumerative or unstructured search} oracle, quantum search speedups are at most quadratic~\cite{bbbv1997}. So their realistic contribution to these problems is likely a Grover-esque quadratic speedup of the branch-and-bound step, but current fault-tolerant overheads erase such speedups~\cite{hoefler2023}. Security constrained OPF inherits these conclusions by a containment argument.

Ironically, efficient \emph{classical} approaches to AC-OPF, such as~\cite{madani2015, jabr2012}, exploit the small treewidth structure that dooms quantum linear algebra. While tempting, hybrid schemes likely do not escape either. Variational solvers still estimate loss by repeated preparation and measurement, so \sat{repeated measurement costs} persist and $\kappa$-dependence returns through the optimization landscape~\cite{dalzell2023}. A decomposition scheme assigning UC's combinatorial step to a quantum device, such as an annealer, \sat{concedes} the continuous subproblems\textemdash which may have fast classical algorithms\textemdash for a speedup at most quadratic. Hybridization is worthwhile if the quantum module avoids both linear system solving and classical readout; it \textit{relocates} rather than removes computational barriers, and requires closer scrutiny for power engineering applications.

\section{Outlook: The Right Classical Baseline?}
\label{sec:outlook}

This letter showed that common electric grid structures can make quantum power flow (QPF) provably slow in areas where classical methods are already fast. Table~\ref{tab:conditions} documents this structure in realistic test cases. These numerical results and all Lean 4 proof formalizations are open source. The Lean 4 code also compiles to an executable C program that certifies whether a test case exhibits QPF obstructions.

How power engineering can best benefit from new computational paradigms remains an open question. \sat{Quantum-inspired and randomized} solvers exist today~\cite{kyng2016,kyng2026} and remain largely untapped in power systems, warranting further study.

\begin{table}[t]
\caption{Grid structure conditions; $\checkmark$ indicates the condition applies.}
\label{tab:conditions}
\centering
\renewcommand{\arraystretch}{1.1}
\setlength{\tabcolsep}{3.2pt}
\begin{tabular}{lcccc}
\hline
Case & $n$ & Near Plan. & Sep. $s$ & $\operatorname{tw}\leq U$ \\
\hline
Texas7k~\cite{birchfield2017} & 6,717 & $\checkmark$ & $\checkmark$~26 & $\checkmark$~46 \\
CATS~\cite{taylor2024cats} & 8,870 & $\checkmark$ & $\checkmark$~6 & $\checkmark$~22 \\
PEGASE13k~\cite{Josz2016ACPF} & 13,659 & $\checkmark$ & $\checkmark$~15 & $\checkmark$~34 \\
Midwest24k~\cite{birchfield2017} & 23,643 & $\checkmark$ & $\checkmark$~51 & $\checkmark$~82 \\
GOC30k~\cite{birchfield2017} & 30,000 & $\checkmark$ & $\checkmark$~33 & $\checkmark$~67 \\
EPIGRIDS78k~\cite{snodgrass2021tractable} & 78,478 & $\checkmark$ & $\checkmark$~57 & $\checkmark$~108 \\
\hline
\end{tabular}
\vspace{-1em}
\end{table}

\bibliographystyle{ieeetr}
{\footnotesize
\bibliography{Refs}
}

\end{document}